%% file: paper.tex
\documentclass[sigconf]{acmart}
\usepackage{amsthm}
\usepackage{amsmath,amsfonts}
\usepackage{algorithm}
\usepackage{algorithmic}
\usepackage{subfigure}
\usepackage{graphicx}
\usepackage{textcomp}
\usepackage{xcolor}

\usepackage{hyperref}
\usepackage{xspace}
\usepackage{threeparttable}
\usepackage{booktabs} 
\usepackage{tabularx}
\usepackage{multirow}
\usepackage{multicol}
\usepackage{longtable}
\usepackage{mdwlist}
\usepackage{enumitem}
\usepackage{footnote}
\usepackage{tablefootnote}
\usepackage{footnote}
\usepackage[bottom]{footmisc}

\xspaceaddexceptions{]\}}

\usepackage{url}
\newcommand{\method}{{DivMF}\xspace}
\newcommand{\methodfull}{Diversely Regularized Matrix Factorization\xspace}

\newtheoremstyle{problemstyle}  
        {3pt}                                               
        {3pt}                                               
        {\normalfont}                               
        {}                                                  
        {\bfseries\itshape}                 
        {\normalfont\bfseries:}         
        {.5em}                                          
        {}                                                  
\theoremstyle{problemstyle}

\newcommand{\rom}[1]{\uppercase\expandafter{\romannumeral #1\relax}}

\setcopyright{none}
\renewcommand\footnotetextcopyrightpermission[1]{} 
\begin{document}
\title{
	Diversely Regularized Matrix Factorization for Accurate and Aggregately Diversified Recommendation
}

\author{Jongjin Kim}
\affiliation{
	\institution{Seoul National University}
	\country{Seoul, Korea}
}
\email{j2kim99@snu.ac.kr}

\author{Hyunsik Jeon}
\affiliation{
	\institution{Seoul National University}
	\country{Seoul, Korea}
}
\email{jeon185@snu.ac.kr}

\author{Jaeri Lee}
\affiliation{
	\institution{Seoul National University}
	\country{Seoul, Korea}
}
\email{jlunits2@snu.ac.kr}

\author{U Kang}
\affiliation{
	\institution{Seoul National University}
	\country{Seoul, Korea}
}
\email{ukang@snu.ac.kr}


\begin{abstract}
When recommending personalized top-$k$ items to users, how can we recommend the items diversely to them while satisfying their needs?
Aggregately diversified recommender systems aim to recommend a variety of items across whole users without sacrificing the recommendation accuracy.
They increase the exposure opportunities of various items, which in turn increase potential revenue of sellers as well as user satisfaction.
However, it is challenging to tackle aggregate-level diversity with a matrix factorization (MF), one of the most common recommendation model,
since skewed real world data lead to skewed recommendation results of MF.

In this work, we propose \method (\methodfull), a novel matrix factorization method for aggregately diversified recommendation.
\method regularizes a score matrix of an MF model to maximize coverage and entropy of top-$k$ recommendation lists to aggregately diversify the recommendation results. 
We also propose an unmasking mechanism and carefully designed mini-batch learning technique 
for accurate and efficient training. 
Extensive experiments on real-world datasets show that \method achieves the state-of-the-art performance in aggregately diversified recommendation.


\end{abstract}
\pagestyle{plain}
\maketitle

\input{010introduction}

\input{020preliminary}

\input{030proposed}

\input{040experiment}
\input{050related}

\input{060conclusion}

\clearpage

\bibliographystyle{ACM-Reference-Format}
\bibliography{paper}

%

\end{document}

%% file: 010introduction.tex
\begin{figure}[t]
\centering
\subfigure[Users' preferences]{\includegraphics[width=0.15\textwidth]{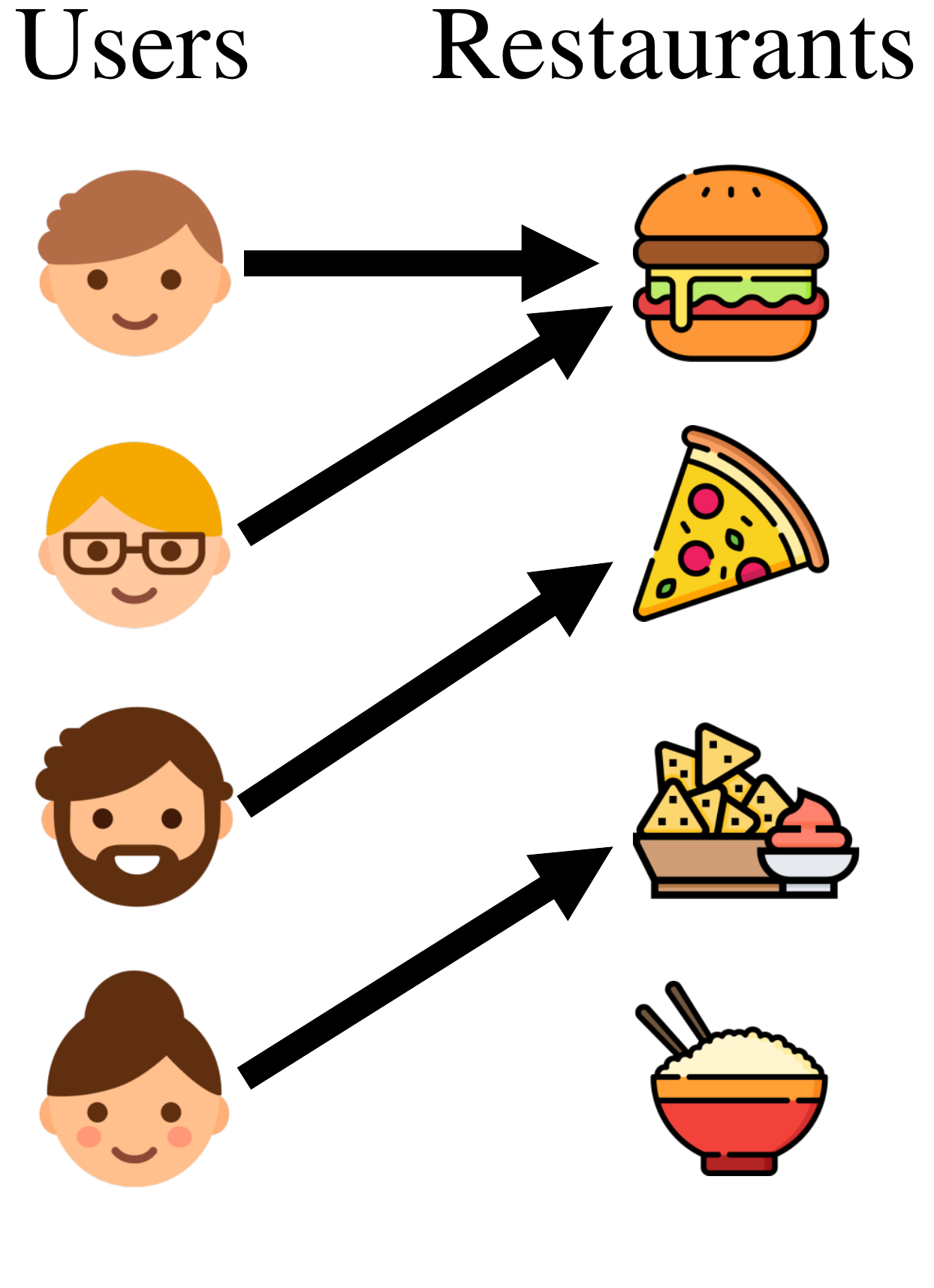}\label{fig:illustration1}}
\subfigure[Three different recommendation results]{\includegraphics[width=0.31\textwidth]{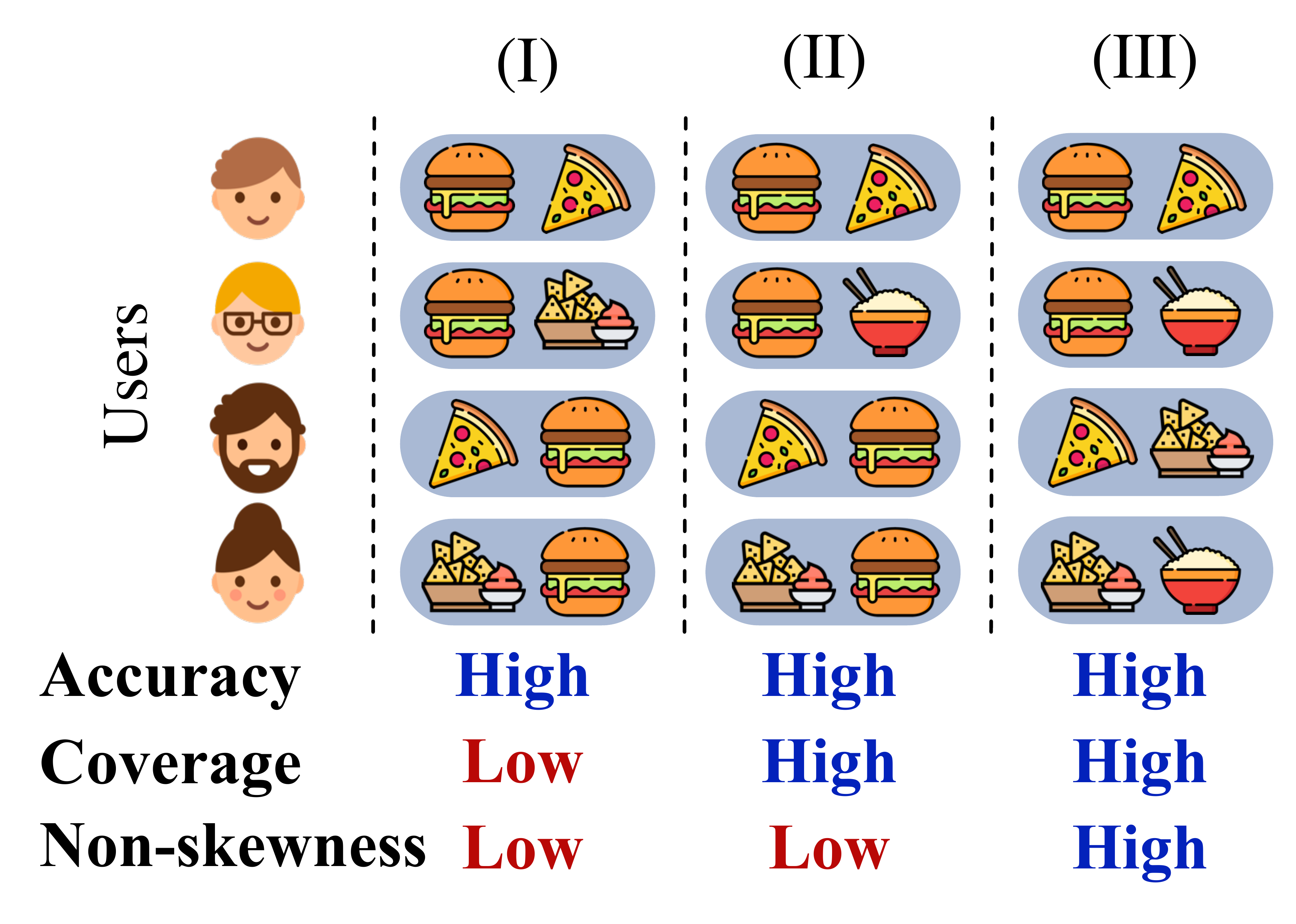}\label{fig:illustration2}}
\caption{
Comparison of three different recommendation results.
The aggregate-level diversities (coverage and non-skewness) in the results (\rom{1}), (\rom{2}), and (\rom{3}) are significantly different, although the accuracies are the same.
Aggregately diversified recommendation (our goal) aims to achieve the result (\rom{3}) where coverage and non-skewness are high with a comparable accuracy.
}
\label{fig:illustration}
\end{figure}

\section{Introduction}
\label{sec:introduction}

\textit{When recommending personalized top-$k$ items to users, how can we recommend the items diversely while satisfying their needs?}
Customers heavily rely on recommender systems to choose items due to the flood of information nowadays.
Thus, it is ideal to expose as many items as possible to users to maximize potential revenue of sales platforms.
This is even more crucial if items have limited capacities such as in restaurants or movie tickets,
since excessive recommendation on a few popular items would not be helpful to increase total revenue of a system if those items' capacities are exceeded.
Hence, aggregate-level diversity in recommender system is important to increase system owners' profit in those situations.

Aggregate-level diversity requires that the recommendation results are of high coverage and low skewness;
coverage indicates the proportion of recommended items among all items, and skewness indicates the degree of unfairness in the frequencies of recommended items.
Both coverage and non-skewness need to be considered in order for all items have fair chances to be exposed.
Figure~\ref{fig:illustration} demonstrates the coverage and non-skewness of three different recommendation results.
Figure~\ref{fig:illustration1} illustrates the ground-truth preferences of users and Figure~\ref{fig:illustration2} exemplifies three different recommendation results.
In the figure, result (\rom{1}) achieves high recommendation accuracy, but low coverage and non-skewness;
all users are recommended an item that they prefer, but some items do not appear on the recommendation or infrequently appear compared to other items.
Result (\rom{2}) achieves high coverage as well as high recommendation accuracy;
note that all the four items are recommended.
Result (\rom{3}) further achieves high non-skewness, since all items are recommended twice.
Note that results (\rom{1}), (\rom{2}), and (\rom{3}) all achieve high individual-level diversity,
since they all recommend two different items for each user. 
However, only result (\rom{3}) achieves a high aggregate-level diversity, recommending
each item by the same amount, maximizing the potential revenue of sales platforms.

Matrix factorization (MF)~\cite{SalakhutdinovM07,koren2009matrix} is the most widely used collaborative filtering method due to its powerful scalability and flexibility~\cite{ko2022survey}.
However, the traditional MF has a limitation in achieving high aggregate-level diversity on real-world data because it is vulnerable to the skewness of data.
Figure~\ref{fig:skewness} shows the skewness of Epinions dataset (see Section~\ref{subsec:experimental_setup} for details), a real-world e-commerce dataset, which includes about 664k interactions made by 40k users and 139k items;
x-axis indicates the popularity rank of items, and y-axis indicates the number of interactions.
As shown in the figure, the interactions are highly skewed to a few items.
For instance, 10\% of items with high interaction ranks occupy 65\% of the total interactions resulting in a long-tail distribution.
MF model tends to give large embedding vectors for popular items so that they have advantages in top-$k$ recommendation.
Hence, carefully designed approach is required to prevent the recommendations from being skewed to a few items in order to achieve high aggregate-level diversity~\cite{ParkT08}.
To overcome this problem, previous works on aggregately diversified recommendation rerank the recommendation lists or recommendation scores of a given MF model~\cite{AdomaviciusK12, AdomaviciusK14, dong2021user, KarakayaA18, MansouryAPMB20}.
However, these approaches does not give a good diversity
since they focus only on post-processing the results of MF.
Thus, it is desired to deal with aggregate-level diversity in the training process of MF to achieve both high accuracy and diversity.

\begin{figure}[t]
\vspace{-5mm}
\includegraphics[width=0.4\textwidth]{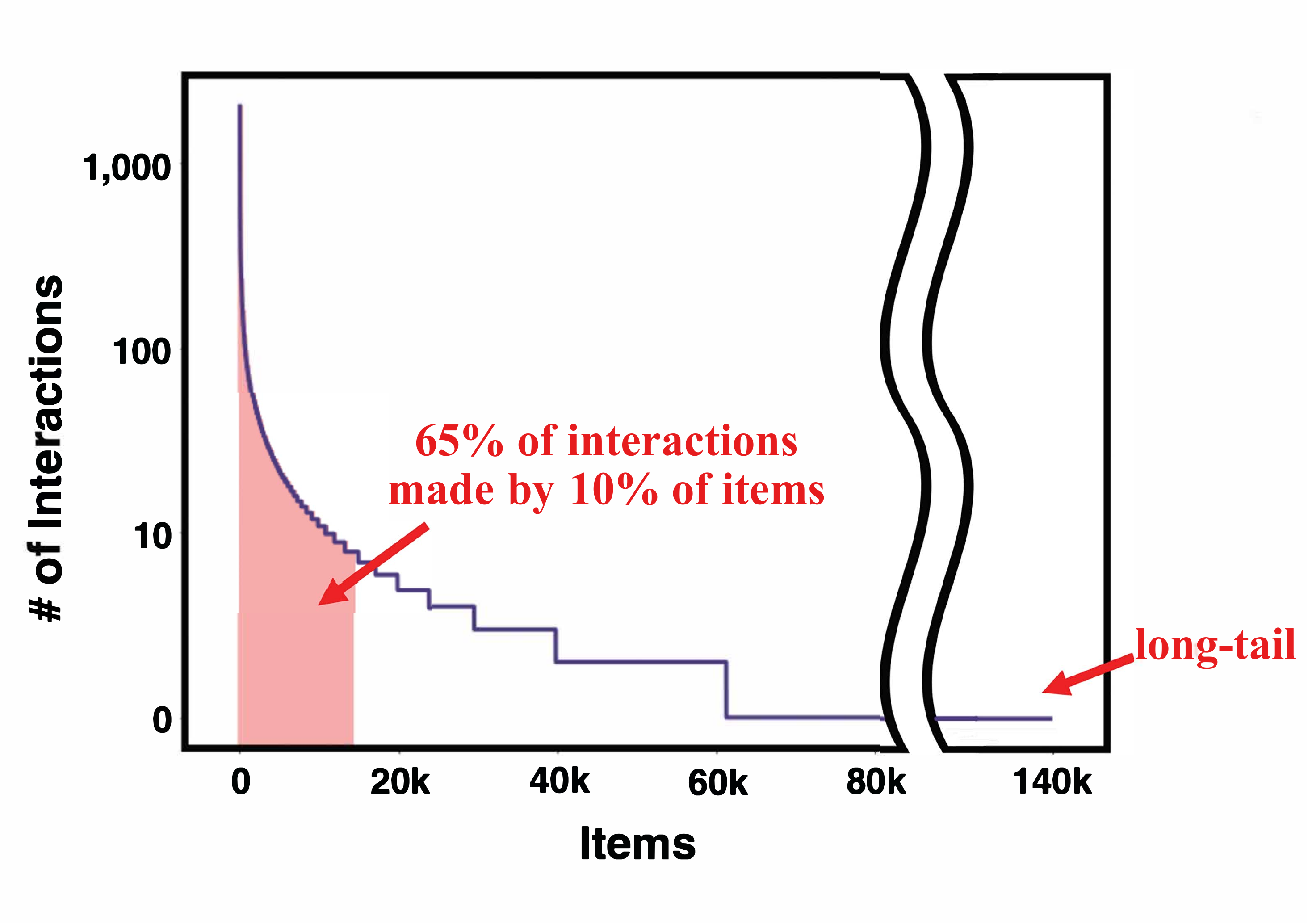}
\vspace{-5mm}
\caption{Skewness of Epinions dataset. The dataset includes 664,823 interactions made by 40,163 users and 139,738 items. 65\% of total interactions are made by 10\% of items.}
\label{fig:skewness}
\end{figure}

In this work, we propose {\methodfull} {(\method)}, a novel approach for aggregately diversified recommendation.
{\method} regularizes a recommendation model in its training process so that more diverse items appear uniformly on top-$k$ recommendations.
{\method} contains two separate regularizers: coverage and skewness regularizers.
The coverage and skewness regularizers effectively maximize the coverage and entropy of recommendation results, respectively.
Both regularizers consider the item occurrences in top-$k$ recommendation list.
This allows the model to achieve the optimal aggregate-level diversity in the training process. 
We also propose an unmasking mechanism and mini-bath learning technique 
for accurate and fast learning of \method.

Our contributions are summarized as follows:
\begin{itemize}
	\item \textbf{Method.}
		We propose {\method}, a method for aggregately diversified recommendation. {\method} provides a new way to accurately and efficiently optimize a recommender to achieve both high accuracy and aggregate-level diversity for top-$k$ recommendation.
	\item \textbf{Theory.}
		We theoretically prove that {\method} provides an optimal solution to maximize the aggregate-level diversity in top-$k$ recommendation by analyzing how our proposed diversity regularizer improves the coverage and non-skewness.
	\item \textbf{Experiments.}
		Extensive experiments show that {\method} achieves up to 34.7\% higher aggregate-level diversity in the similar level of accuracy,
		and up to 27.6\% higher accuracy in the similar level of aggregate-level diversity in personalized top-$k$ recommendation compared to the best competitors,
		resulting in the state-of-the-art performance (see Figure~\ref{FIG:performance}).
\end{itemize}

In the rest of our paper, we provide the preliminary in Section~\ref{sec:preliminary}, introduce our proposed method in Section~\ref{sec:proposed}, evaluate the proposed method and competitors on real-world datasets in Section~\ref{sec:experiments}, review related works in Section~\ref{sec:related}, and conclude this work in Section~\ref{sec:conclusion}.
Table \ref{tab:sym} lists the frequently used symbols in this paper.
The code and datasets are available at \url{https://github.com/DivMF/DivMF}.

\begin{savenotes}
\begin{table}
	\centering
	\caption{Frequently used symbols.}
	\begin{tabular}{c|l}
\toprule
\textbf{Symbol} & \textbf{Definition} \\
\midrule
$\mathbb{U}$ & set of all users \\
$\mathbb{I}$ & set of all items \\
$u$ & index of a user \\
$i$ & index of an item \\
$r_{ui}$ & ground-truth interaction between $u$ and $i$ \\
$\hat{r}_{ui}$ & recommendation score between $u$ and $i$\\
$\mathbf{R}=[r_{ui}]$ & ground-truth user-item interaction matrix \\
$\mathbf{\hat{R}}=[\hat{r}_{ui}]$ & recommendation score matrix\\
$k$ & length of a recommendation list for a user \\
$\mathbb{L}_k(u)$ & recommendation list of length $k$ for $u$ \\
\toprule
	\end{tabular}
	\label{tab:sym}
\end{table}
\end{savenotes}


%% file: 020preliminary.tex
\section{Preliminary}
\label{sec:preliminary}

\subsection{Aggregately Diversified Recommendation}
\label{subsec:div_rec}
In recent years, diversification has attracted increasing attention in recommendation research since it has been recognized as an important requirement for improving user satisfaction~\cite{WuLMZT19}.
The diversity in recommender systems is classified into two levels: individual-level and aggregate-level.
Individual-level diversity independently considers the diversity in recommendation results of individuals, whereas aggregate-level diversity considers the diversity in the integrated recommendation results of all users.
Maximizing individual-level diversity provides diverse choices to each user while maximizing aggregate-level diversity widens the pool of total items recommended to users.
In this paper, we focus on aggregately diversified recommender systems. 

The problem of maximizing aggregate-level diversity is defined formally as follows.
Given a sparse user-item interaction matrix $\mathbf{R}\in\mathbb{R}^{|\mathbb{U}|\times|\mathbb{I}|}$ where $\mathbb{U}$ and $\mathbb{I}$ are sets of users and items, respectively,
the goal of aggregately diversified recommendation is to predict a dense recommendation score matrix $\mathbf{\hat{R}}$ that maximizes recommendation accuracy and aggregate-level diversity, so that the recommendation satisfies users' taste and is not skewed to few items.

The aggregate-level diversity of recommendation is evaluated by the following three metrics.
\begin{itemize}
	\item\textbf{Coverage.} Coverage measures how many different items a recommendation result contains from the whole items.
	It is defined as follows:
	\begin{equation}
		Coverage@k = \frac{\begin{vmatrix}\mathsf{U}_{u\in{\mathbb{U}}}\mathbb{L}_k(u)\end{vmatrix}}{\begin{vmatrix}\mathbb{I}\end{vmatrix}},
	\end{equation}
	where $k$ is the number of items recommended and $\mathbb{L}_k(u)$ is the set of top-$k$ recommended items for user $u$. $\mathbb{U}$ and $\mathbb{I}$ are sets of users and items, respectively.
	Coverage ranges from 0 to 1, and a high value indicates that the recommendation results cover many items, which satisfies a desired property for aggregately diversified recommendation.
	\item\textbf{Entropy.}
	The entropy of a random variable represents the uncertainty of its possible outcomes.
	High uncertainty means that the random variable has diverse outcomes.
	It is defined as follows:
	\begin{equation}
		Entropy = - \sum_{i\in{\mathbb{I}}}p(i)\log{p(i)},
	\end{equation}
	where $\mathbb{I}$ is a set of items.
	$p(i)=\frac{f(i)}{\sum_{j\in\mathbb{I}}{f(j)}}$ where $f(i)$ indicates the frequency of item $i$ in the recommendation results for whole users.
	A high value of entropy in recommendation results indicates that the probability of items appearing in the recommendation results is not skewed, which satisfies a desired property for aggregately diversified recommendation.
	\item\textbf{Gini index.}
	Gini index measures the inequality of item frequencies in recommendation results.
	It is defined as follows:
	\begin{equation}
		Gini = \frac{1}{|\mathbb{I}|-1}\sum_{j=1}^{|\mathbb{I}|}(2j-|\mathbb{I}|-1)p_j,
	\end{equation}
	where $p_j$ is the $j$-th least value in $\{p(i)|i\in\mathbb{I}\}$.
	A Gini index ranges from 0 to 1.
	A low Gini index indicates that the recommendation results show a non-skewed distribution of item frequencies, which is a desired property for aggregately diversified recommendation.
	
\end{itemize}

%% file: 030proposed.tex
\section{Proposed Method}
\label{sec:proposed}

In this section, we propose \methodfull (\method), a method for accurate and aggregately diversified recommendation.

\subsection{Overview}
\label{subsec:overview}
We address the following challenges to achieve high performance of aggregately diversified recommendation:
\begin{itemize}
	\item \textbf{Coverage maximization.}
	Matrix factorization (MF) is prone to obtaining top-$k$ recommendations with low coverage where only a few items are recommended.
	How can we train MF to recommend every item at least once?
	\item \textbf{Non-skewed frequency.}
	MF is liable to achieving skewed top-$k$ recommendations.
	How can we train MF to recommend all items with similar frequencies?
	\item \textbf{Non-trivial optimization.}
	It is difficult to simultaneously handle both accuracy and diversity which are disparate criteria.
	How can we train MF to optimize the accuracy and diversity?
	
\end{itemize}

The main ideas of \method are summarized as follows:
\begin{itemize}
	\item \textbf{Coverage regularizer.}
	%
	The coverage regularizer evenly balances the recommendation scores at the item-level, enabling us to recommend each item to at least one user.
	\item \textbf{Skewness regularizer.}
	%
	The skewness regularizer equalizes all the recommendation scores to assist the coverage regularizer to make the model recommends all items the same amounts of times.
	\item \textbf{Careful training.}
	We perform sequential training, not alternating training, which first focuses on accuracy and then diversity for stable training.
	We also propose an unmasking mechanism and mini-batch learning for effective and efficient training.
\end{itemize}


\subsection{Definition of Diversity Regularizer}
\label{subsec:diversity}
Given a sparse interaction matrix $\mathbf{R}\in\mathbb{R}^{|\mathbb{U}|\times|\mathbb{I}|}$, our goal is to predict a dense recommendation score matrix $\mathbf{\hat{R}}$ that maximizes recommendation accuracy and aggregate-level diversity.
In this section, we propose two regularizers, coverage and skewness regularizers, to maximize the aggregate-level diversity.

\subsubsection{Coverage Regularizer}
We construct a coverage regularizer to assure that every item is recommended to at least one user.
To achieve this purpose, we note that the sum of the scores of an item would represent the occurrence of the item in top-$k$ recommendation if we eliminate the lowest $|\mathbb{I}|-k$ scores of each user.
Thus, the regularizer that equalizes the sum of scores of top-$k$ items would accomplish the goal.

Assume that $\mathbf{\hat{R}}=[\hat{r}_{ui}]$ is the recommendation score matrix where
$\hat{r}_{ui}$ is a dot product between user $u$'s embedding and item $i$'s embedding.
For $u\in \mathbb{U}$, consider $\mathbf{S}=[s_{ui}]$ where $\mathbf{S}_u = softmax(\mathbf{\hat{R}_u})$, which means $(s_{u1}, ... , s_{u|\mathbb{I}|}) = softmax(\hat{r}_{u1}, ... , \hat{r}_{u|\mathbb{I}|})$.
Let a function $top(v, k)$ takes a vector $v=(v_1, v_2, ... , v_n)$ and a positive integer $k$ as inputs, and returns $v'=(v'_1, v'_2, ... , v'_n)$ where $v'_i$ is defined as follows:
\begin{equation*}
	v'_i=
	\begin{cases}
		v_i & \text{($|\{j| v_j> v_i \text{ for } v_j\in\{v_1, v_2, \dots, v_n\}\}|< k$)}\\
		0 &\text{(otherwise)},
	\end{cases}
\end{equation*}
which means to keep only top $k$ largest elements of $v$.
Let matrix $\mathbf{T} = [t_{ui}]$ where $\mathbf{T}_u = (t_{u1}, ... , t_{u|\mathbb{I}|}) = top(\mathbf{S}_u, k)$.
Note that the nonzero element of $\mathbf{T}$ in $u$th row and $i$th column implies that the top-$k$ recommendation list of user $u$ includes item $i$.
Then, the coverage regularizer $Reg_{cov}$ is defined as follows:

\begin{align*}
	Reg_{cov} &= -\log (\prod_{i\in \mathbb{I}}\sum_{u\in \mathbb{U}} t_{ui}) = -\sum_{i\in \mathbb{I}}\log(\sum_{u\in \mathbb{U}}t_{ui}).
\end{align*}

This regularizer is useful for maximizing coverage, as shown in the following theorem. 

\begin{theorem}
	\label{thm:coverage}
	If $Reg_{cov}$ is minimized, then coverage is maximized.
\end{theorem}

\begin{proof}
		$\sum_{u\in \mathbb{U}}t_{ui}\le\sum_{u\in \mathbb{U}}s_{ui}$ for all $i\in\mathbb{I}$ since $0\le t_{ui} \le s_{ui}$ for all $u\in\mathbb{U}$ and $i\in\mathbb{I}$. Thus, using the fact that $\sum_{i\in\mathbb{I}}s_{ui}=1$ for all $u\in\mathbb{U}$,
		\begin{equation*}
			\sum_{i\in \mathbb{I}}\sum_{u\in \mathbb{U}} t_{ui} \le \sum_{i\in \mathbb{I}}\sum_{u\in \mathbb{U}} s_{ui} = \sum_{u\in \mathbb{U}}\sum_{i\in \mathbb{I}} s_{ui} = |\mathbb{U}|.
		\end{equation*}
		We thus obtain
		\begin{equation*}
			\exp(-Reg_{cov}) = \prod_{i\in \mathbb{I}}\sum_{u\in \mathbb{U}} t_{ui} \le (\frac{|\mathbb{U}|}{|\mathbb{I}|})^{|\mathbb{I}|},
		\end{equation*}
		from arithmetic geometric mean inequality.
		Equality holds if and only if
		\begin{equation}
			\label{equality_D}
			\forall i, \sum_{u\in \mathbb{U}} t_{ui}=|\mathbb{U}|/|\mathbb{I}|.
		\end{equation}
		In this case, every column of $\mathbf{T}$ has at least one nonzero element since $|\mathbb{U}|/|\mathbb{I}|>0$ and $t_{ui}\ge 0$.
		Thus, every item is included in at least one user's top-$k$ recommendation list, so the coverage is $1$.
\end{proof}

\subsubsection{Skewness Regularizer}
Although the condition to minimize the coverage regularizer guarantees the coverage of model to be $1$,
this does not guarantee other diversity metrics such as entropy and Gini index to be improved because the skewed frequency distribution of items in a recommendation result is not prevented.
For example, assume that $(t_{11}, t_{21}, ... , t_{|\mathbb{U}|1})=(\frac{1}{2}, \frac{1}{2}, 0, 0, ... , 0)$ and
$(t_{12}, t_{22}, ... , t_{|\mathbb{U}|2})=(\frac{1}{3}, \frac{1}{3}, \frac{1}{3}, 0, 0, ... , 0)$.
In this case, $\sum_{u\in\mathbb{U}}t_{u1}=\sum_{u\in\mathbb{U}}t_{u2}$ but the item $1$ is recommended twice while the item $2$ is recommended three times;
note that a nonzero $t_{ui}$ indicates that the top-$k$ recommendation list of user $u$ includes item $i$.
In other words, it is possible to meet Equation~\eqref{equality_D} if the number of nonzero elements in each column of $\mathbf{T}$ is not equal to each other, since the value of each nonzero element could vary.
This means each item is recommended a different number of times in recommendation results.

To address the challenge, we propose a skewness regularizer to make values of positive $t_{ui}$ to be close to each other.
In this way, the sum of $t_{ui}$ would indicate the frequency of item $i$ on recommendation results, so the coverage regularizer would also optimize the skewness in recommendation lists.
Let $\mathbf{T'}=[t'_{ui}]$ be a row-wise normalized $\mathbf{T}$ which means $t'_{ui} = {t_{ui}}/{\sum_{j\in \mathbb{I}}t_{uj}}$.
The skewness regularizer $Reg_{skew}$ is defined as follows:
\begin{align*}
	Reg_{skew} &= \sum_{u\in \mathbb{U}}\sum_{i\in \mathbb{I}} t'_{ui}\log{t'_{ui}}\\
	&= -\sum_{u\in \mathbb{U}}entropy(\mathbf{T}_u).
\end{align*}

Then, this regularizer satisfies the Lemma \ref{equality_E}.

\begin{lemma}
	\label{equality_E}
	$Reg_{skew}$ is minimized if and only if all nonzero elements of each row of $\mathbf{T}$ are equal.
\end{lemma}
\begin{proof}
	Since each row of $\mathbf{T}$ is independent, $Reg_{skew} =$\\$ -\sum_{u\in \mathbb{U}}entropy(\mathbf{T}_u)$ is minimized if and only if all of $entropy(\mathbf{T}_u)$ are maximized respectively.
	Without loss of generality, let us find the condition to maximize $entropy(\mathbf{T}_1)$.	
	$\mathbf{T}_1$ has at most $k$ nonzero elements because of its definition, and $entropy(\mathbf{T}_1)$ is maximized if and only if top-$k$ elements of $\mathbf{S}_1$ are equal.
	Thus, $Reg_{skew}$ is minimized if and only if each row's top-$k$ elements of $\mathbf{S}$ are equal.
\end{proof}

\subsubsection{Diversity Loss Function}
Finally, we define a loss function for aggregate-level diversity in \method as follows:
\begin{align*}
	\mathcal{L}_{div}(\mathbf{\hat{R}}) &= Reg_{cov} + Reg_{skew}.
\end{align*}
This loss function satisfies the Theorem \ref{thm:frequency}.

\begin{theorem}
	\label{thm:frequency}
If $\mathcal{L}_{div}(\mathbf{\hat{R}})$ is minimized, then
	Gini index is zero, entropy is maximized, and coverage is maximized.
\end{theorem}
\begin{proof}
	$\mathcal{L}_{div}(\mathbf{\hat{R}}) = Reg_{cov} + Reg_{skew}$ is minimized if and only if $Reg_{cov}$ and $Reg_{skew}$ are both minimized simultaneously if such condition is possible.
	The condition to minimize $Reg_{cov}$ is the Equation~\eqref{equality_D},
	and the condition to minimize $Reg_{skew}$ is given in the Lemma~\ref{equality_E}.
	To meet the Equation~\eqref{equality_D}, all elements of $\mathbf{S}$ should be zero except for each row's top-$k$ elements,
	which is equivalent to $\mathbf{T}=\mathbf{S}$,
	and the sum of elements in each column of $\mathbf{S}$ should be equal to $\frac{|\mathbb{U}|}{|\mathbb{I}|}$.
	Meanwhile, top-$k$ elements of each row of $\mathbf{S}$ should be equal to each other to meet the condition given in the Lemma~\ref{equality_E}.
	Hence, the conditions to minimize both regularizers are that each row of $\mathbf{S}$ contains $|\mathbb{I}|-k$ zeros and $k$ nonzero elements with value of $\frac{1}{k}$, and each column of $\mathbf{S}$ contains $\frac{|\mathbb{U}|k}{|\mathbb{I}|}$ nonzero elements,
	which are possible to meet.
	Thus, these are also the conditions to minimize $\mathcal{L}_{div}(\mathbf{\hat{R}})$.
	In this case, every item appears with equal frequency since every column of $\mathbf{T}$ contains $\frac{|\mathbb{U}|k}{|\mathbb{I}|}$ nonzero elements.
	Therefore, Gini index is zero, entropy would be maximized, and coverage would be maximized if $\mathcal{L}_{div}(\mathbf{\hat{R}})$ is minimized.
\end{proof}

\subsection{Model Training}

\subsubsection{Objective Function and Training Algorithm}
\label{subsec:objective}
In order to maximize accuracy and aggregate-level diversity of recommendation results simultaneously, we propose the following objective function.
\begin{equation*}
	\mathcal{L}_{total}(\theta;\mathbf{R}) = \mathcal{L}_{acc}(\mathbf{\hat{R}}) + \mathcal{L}_{div}(\mathbf{\hat{R}}),
\end{equation*}
where $\mathcal{L}_{total}(\cdot)$ is the total loss that is to be minimized,
$\mathcal{L}_{acc}(\cdot)$ and $\mathcal{L}_{div}(\cdot)$ are losses for accuracy and aggregate-level diversity, respectively,
$\mathbf{R}$ is the observed interaction matrix,
$\mathbf{\hat{R}}$ is the recommendation score matrix,
and $\theta$ is the parameter to be optimized.
There are a lot of potential $\mathcal{L}_{acc}(\mathbf{\hat{R}})$ functions such as RMSE, MSE, BPR, etc.;
we use BPR loss function since it is known to show the best performance in top-$k$ recommendation~\cite{RendleFGS09}.
Thus,
\begin{equation*}
	\mathcal{L}_{acc}(\mathbf{\hat{R}})
	= -\sum_{u\in\mathbb{U}, (i, j)\in \mathbb{Z}(u)}
	{\log(\frac{1}{1+\exp(-(\mathbf{\hat{R}}_{ui}-\mathbf{\hat{R}}_{uj}))})},
\end{equation*}
where $\mathbb{Z}(u)=\{(i, j)| \mathbf{R}_{ui} = 1, \mathbf{R}_{uj}=0\}$.

A challenge in minimizing the loss $\mathcal{L}_{total}$
is that directly minimizing $\mathcal{L}_{total}$ or optimizing $\mathcal{L}_{acc}$ and $\mathcal{L}_{div}$ in an iterative, alternating fashion leads to a poor performance.
We presume that this problem happens because the gradients of accuracy loss and diversity regularizer cancel each other out.
The accuracy loss tries to increase the gap between recommendation scores of high scored items and low scored items, while the diversity regularizer tries to decrease the gap.
Thus, the net gradient is not large enough to prevent the model from being trapped in bad local optima.

Our idea to avoid this issue is
to train \method model with only accuracy loss until the accuracy converges, and then train the model with the diversity regularizer $\mathcal{L}_{div}(\cdot)$.
In this way, the gradients of accuracy loss and diversity regularizer do not cancel each other out since the optimizer optimizes only one loss at a time.
Algorithm \ref{al:training} shows the overall process of \method.
\method fits the model to generate a valid user-item score prediction matrix $\mathbf{\hat{R}}$ in lines 1 to 4.
Then \method modifies the parameters of model to increase the aggregate diversity in lines 5 to 8.
The model achieves higher diversity and lower accuracy as we increase $n_{ep}$, the number of epochs to optimize the diversity regularizer.

\renewcommand{\algorithmicrequire}{\textbf{Input:}}
\renewcommand{\algorithmicensure}{\textbf{Output:}}
\begin{algorithm}[t]
	\caption{Training \method}
	\begin{algorithmic}[1]
		\REQUIRE User-item interaction matrix $\mathbf{R}$, and accuracy-diversity trade-off hyperparameter $n_{ep}$
		\ENSURE Trained model parameters $\theta$
		\WHILE {(Validation accuracy does not converge)}
			\STATE Calculate the recommendation loss $\mathcal{L}_{acc}(\mathbf{\hat{R}})$
			\STATE Update the parameters $\theta$ to minimize $\mathcal{L}_{acc}(\mathbf{\hat{R}})$
		\ENDWHILE
		\FOR{$[ 1, n_{ep} ]$}
			\STATE Calculate the diversity loss $\mathcal{L}_{div}(\mathbf{\hat{R}})={Reg_{cov}}+Reg_{skew}$
			\STATE Update the parameters $\theta$ to minimize $\mathcal{L}_{div}(\mathbf{\hat{R}})$
		\ENDFOR
	\end{algorithmic}
	\label{al:training}
\end{algorithm}

\subsubsection{Unmasking Mechanism}
\label{subsec:unmasking}
Gradients from $\mathcal{L}_{div}(\mathbf{\hat{R}})$ do not flow directly into unrecommended items since $\mathbf{T}$ masks $|\mathbb{I}|-k$ items with the lowest scores in $\mathbf{S}$ of each user.
Therefore, a straightforward gradient descent with $\mathcal{L}_{div}(\mathbf{\hat{R}})$
does not actively find new items for diversity,
optimizing only $k$ item scores initially selected.
%

%
We propose an unmasking mechanism to overcome this problem.
The idea is to keep additional unmasked elements in each row of $\mathbf{S}$  while building $\mathbf{T}$.
In this way, rarely recommended items are more likely to be unmasked.
\method finds new rarely recommended items by a gradient descent with this unmasking mechanism.
\method unmasks a fixed number of the highest scored items other than already recommended items during each iteration of training,
which is the best unmasking scheme as experimentally shown in Section~\ref{subsec:unmask}.

\subsubsection{Mini-Batch Learning}
\label{subsec:batch}
%

We propose a mini-batch learning technique to accelerate the training process of \method.
Mini-batch of \method is a submatrix $\mathbf{\hat{R}_b}\in \mathbb{R}^{r_b\times c_b}$ of $\mathbf{\hat{R}}$ that has $r_b$ rows and $c_b$ columns. 
We assume users and items are independent of other users and items, respectively. 
For the reason, the column wise sum of $t_{ui}$ in $\mathbf{\hat{R}_b}$ would be $r_b/|\mathbb{U}|$ times of that in $\mathbf{\hat{R}}$ on average.
Meanwhile, each row of $\mathbf{\hat{R}_b}$ would contain $c_b/|\mathbb{I}|\times k$ nonzero elements of $\mathbf{T}$ on average. 
Thus, optimizing $\mathcal{L}_{div}(\mathbf{\hat{R}_b})$ under the assumption that the model recommends $c_b/|\mathbb{I}|\times k $ items for each user also optimizes $\mathcal{L}_{div}(\mathbf{\hat{R}})$.
In this way, \method reduces the training time and the memory space requirement compared to directly optimizing $\mathcal{L}_{div}(\mathbf{\hat{R}})$.

%% file: 040experiment.tex
\begin{savenotes}
\begin{table}[t]
	\centering
	\caption{%
		Summary of datasets.
	}

	\begin{threeparttable}
	\resizebox{8.5cm}{!}{
	\begin{tabular}{lrrrr}
		\toprule
		\textbf{Dataset}
			& \textbf{Users}
			& \textbf{Items}
			& \textbf{Interactions}
			& \textbf{Density(\%)}\\
		\midrule
		Yelp-15\footnote{\url{https://www.yelp.com/dataset}}
			& 69,853 & 43,671 & 2,807,606 & 0.0920\\
		Gowalla-15\footnote{\url{https://snap.stanford.edu/data/loc-gowalla.html}}
			& 34,688 & 63,729 & 2,438,708 & 0.1111\\
		Epinions-15\footnote{\url{http://www.trustlet.org/downloaded_epinions.html}}
			& 5,531 & 4,286 & 186,995 & 0.7888\\
		MovieLens-10M\footnote{\url{https://grouplens.org/datasets/movielens/10m/}}
			& 69,878 & 10,677 & 10,000,054 &  1.3403\\
		MovieLens-1M\footnote{\url{https://grouplens.org/datasets/movielens/1m/}}
			& 6,040 & 3,706 & 1,000,209 &  4.4684\\
		\bottomrule
	\end{tabular}
	}
	\end{threeparttable}

	\label{table:datasets}
\end{table}
\end{savenotes}

\section{Experiments}
\label{sec:experiments}

We perform experiments to answer the following questions:
\begin{itemize}
	\item[Q1.] \textbf{Diversity and accuracy (Section \ref{subsec:diversity_accuracy}).} Does \method show high aggregate-level diversity without sacrificing the accuracy of recommendation?
\item[Q2.] \textbf{Regularizer (Section \ref{subsec:ablation_study}).} How do the diversity regularization terms $Reg_{cov}$ and $Reg_{skew}$ of \method help improve the diversity performance of \method?
	\item[Q3.] \textbf{Training algorithm (Section \ref{subsec:exp_algs}).} Does the current design of training algorithm prevent \method from being trapped in bad local optima?
	\item[Q4.] \textbf{Unmasking mechanism (Section \ref{subsec:unmask}).} How does the unmasking mechanism affect the performance of \method?
\item[Q5.] \textbf{Mini-batch learning (Section \ref{subsec:batch_exp}).} Does the mini-batch learning technique maintain the performance of \method?
\end{itemize}

\subsection{Experimental Setup}
\label{subsec:experimental_setup}
We introduce our experimental setup including datasets, evaluation protocol, baseline approaches, evaluation metrics, and the training process.

\textbf{Datasets.}
We use five real-world rating datasets as summarized in Table~\ref{table:datasets}.
We reduce sparse datasets such as Yelp, Epinions, and Gowalla to have only users and items which have interacted with at least $15$ others, which we call as $15$-core preprocessing.
MovieLens-10M and MovieLens-1M datasets~\cite{HarperK15} contain movie rating constructed by the GroupLens research group.
Yelp-15 contains 15-core restaurant rating data collected from a restaurant review site with the same name.
Epinions-15~\cite{KasperMD08} contains 15-core rating data of items purchased by users constructed from a general consumer review site.
Gowalla-15~\cite{ChoSL11} contains 15-core data of a friendship network of users constructed from a location-based social networking website where users share their check-in locations.
We remove the rating scores of Yelp-15, Epinions-15, MovieLens-10M, MovieLens-1M datasets and obtain user-item interaction data which indicate whether the user has rated the item or not.

\textbf{Evaluation Protocol.}
We employ \textit{leave-one-out} protocol~\cite{HeLZNHC17,ChenL0GZ19} where one of each user's interaction instances is removed for testing.
Our task is to recommend the item of the removed instance to each user.
From each listed dataset, one instance is selected from the instance list of each user and separated as an element in our test set.
If the dataset includes timestamp, the latest instance of each user is selected, and if not, test set instances are randomly sampled.

\textbf{Baselines.}
We compare \method with existing methods of aggregately diversified recommendation.
\begin{itemize}
	\item \textbf{Kwon.} Kwon et al.~\cite{AdomaviciusK12} is a reranking based approach which adjusts recommendation scores of items based on their frequencies to achieve aggregate level diversity.
	\item \textbf{Karakaya.} Karakaya et al.~\cite{KarakayaA18} replaces items on recommendation lists with infrequently recommended similar items.
	\item \textbf{Fairmatch.} Fairmatch~\cite{MansouryAPMB20} is a reranking based approach which utilizes a maximum flow problem to find important items.
	\item \textbf{UImatch.} UImatch~\cite{dong2021user} is one of the most recent method of achieving aggregate level diversity.
	This method assigns recommendation capacity to each item and constructs its recommendation list in a greedy manner based on a normalized user-item interaction score.
\end{itemize}

\textbf{Evaluation metrics.}
We evaluate the performance of the methods in two categories: accuracy and diversity. Accuracy metric compares the predicted ranks of recommended items to their true ranks, and diversity metrics evaluate aggregated diversity of the recommendation.
For each experiment, a list of recommendation to each user is created and evaluated by the following metrics.
\begin{itemize}
	\item \textbf{Accuracy.}
		\begin{itemize}			
			\item \textbf{nDCG@$k$.}
			nDCG@$k$ measures the overall accuracy of the recommendation lists when $k$ items are recommended to each user.
			It considers highly ranked items more importantly than lowly ranked items to judge the recommendation accuracy.
			It ranges from 0 to 1, where the value 0 indicates the lowest accuracy and the value 1 represents the highest accuracy.
		\end{itemize}
	\item \textbf{Diversity.}
		\begin{itemize}
			\item \textbf{Coverage@$k$.}  The coverage of total recommendation lists when $k$ items are recommended to each user.
			\item \textbf{Entropy@$k$.} The entropy of total recommendation lists when $k$ items are recommended to each user.
			\item \textbf{Negative Gini index@$k$.} The negative value of the Gini index of total recommendation lists when $k$ items are recommended to each user.
		\end{itemize}
\end{itemize}

\textbf{Training Details.}
%
We first train the MF model until it shows the best accuracy.
Then, we apply each baseline method and \method on the trained MF model, and compare results.
We min-max normalize the recommendation scores for Kwon and Karakaya since they need prediction ratings in finite scale.
We use reverse prediction scheme and set $T_H=0.8, T_R=0.9$ for Kwon.
We use hyperparameters $t=[30, 50, 75, 100]$ and $\alpha=0.5$ for FairMatch.
We unmask $50$ items in Epinions-15 dataset,
$100$ items in ML-1M/ML-10M datasets,
and $500$ items in Gowalla-15/Yelp-15 datasets.
We apply the mini-batch learning technique (Section 3.3.3) to train \method in ML-10M, Gowalla-15, and Yelp-15 datasets with batch size $r_b=5000$ and $c_b=5000$.
All the models are trained with Adam optimizer with learning rate $0.001$, $l_2$ regularization coefficient $0.0001$, $\beta_1=0.9$, and $\beta_2=0.999$.
We set $k=5$ to generate recommendation lists.

\subsection{Diversity and Accuracy (Q1)}
\label{subsec:diversity_accuracy}

\begin{figure*}[ht]
	\centering
	\subfigure{\includegraphics[width=0.65\textwidth]{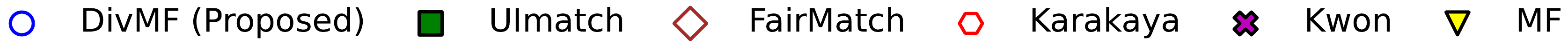}}\\\vspace{-3mm}
	\setcounter{subfigure}{0}
	\subfigure
	{\includegraphics[width=1\textwidth]{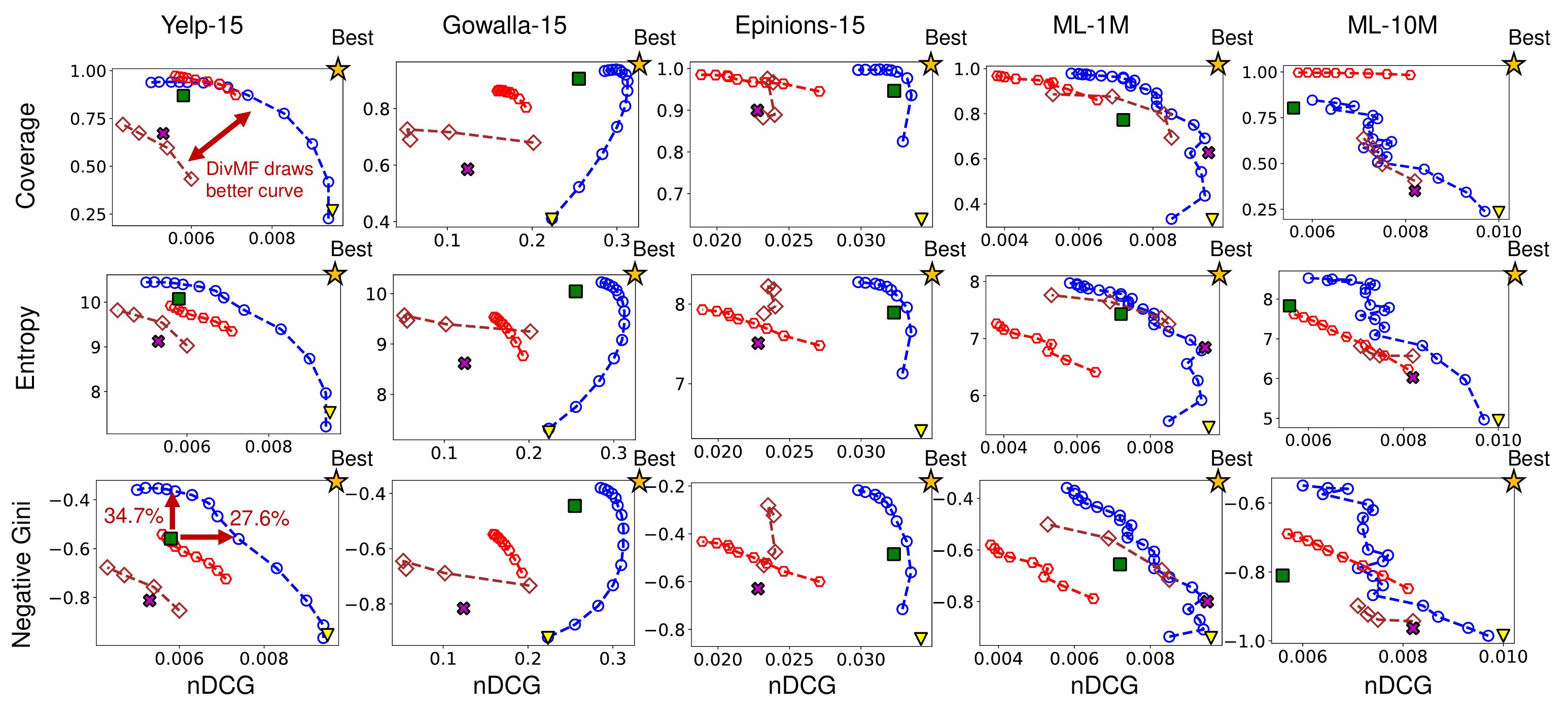}}
	\vspace{-4mm}
	\caption{
		Accuracy-diversity trade-off curves of top-$5$ recommendations on five real-world datasets.
		\method draws the best trade-off curves which are closer to the best points.
		In other words, \method achieves the highest aggregate-level diversity while sacrificing minimal accuracy.
	}
	\label{FIG:performance}
\end{figure*}

We observe the change of accuracies and diversities of \method and the competitors in five real-world datasets.
Figure~\ref{FIG:performance} shows the results of each competitor in each dataset.
We have two main observations.
First, \method achieves better accuracy in similar levels of diversity and better diversity in similar levels of accuracy compared to other baselines, including the previous state-of-the-art methods for aggregately diversified recommendation.
Second, \method shows the best overall diversity considering the balance of coverage, entropy, and Gini index.
Karakaya shows the best coverage on ML-10M, because random walk easily finds new items in a dense dataset.
However, Karakaya's recommendation results are less diversified than \method in terms of entropy and Gini index.
This means that previously rarely recommended items in Karakaya's results are recommended too few times to meaningfully increase the aggregate-level diversity.
Unlike Karakaya, \method also improves entropy and Gini index, which indicate that previously rarely recommended items are also frequently recommended.
Overall, our proposed \method is closest to the `best' point with the maximum accuracy and diversity.
%

\subsection{Regularizer (Q2)}
\label{subsec:ablation_study}

\begin{figure}[h]
	\centering
	 \subfigure{\includegraphics[width=0.9\columnwidth]{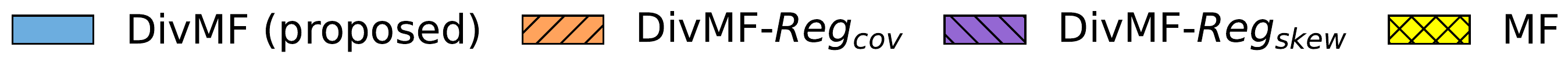}}\\\vspace{-3mm}
	 \setcounter{subfigure}{0}
	 \subfigure[Coverage]{\includegraphics[width=0.32\columnwidth]{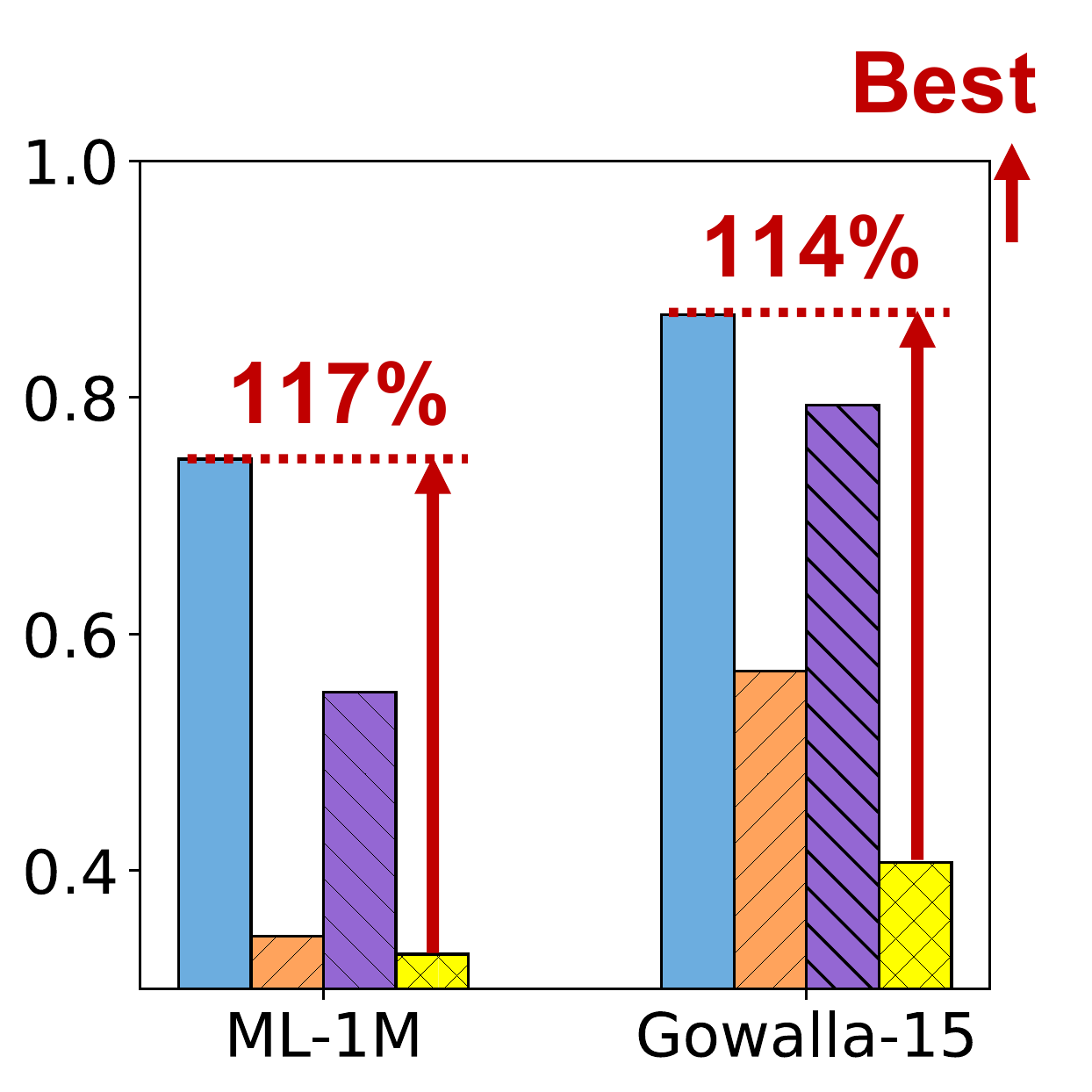}\label{fig:ablation_ml1m_cov}}
	 \subfigure[Entropy]{\includegraphics[width=0.31\columnwidth]{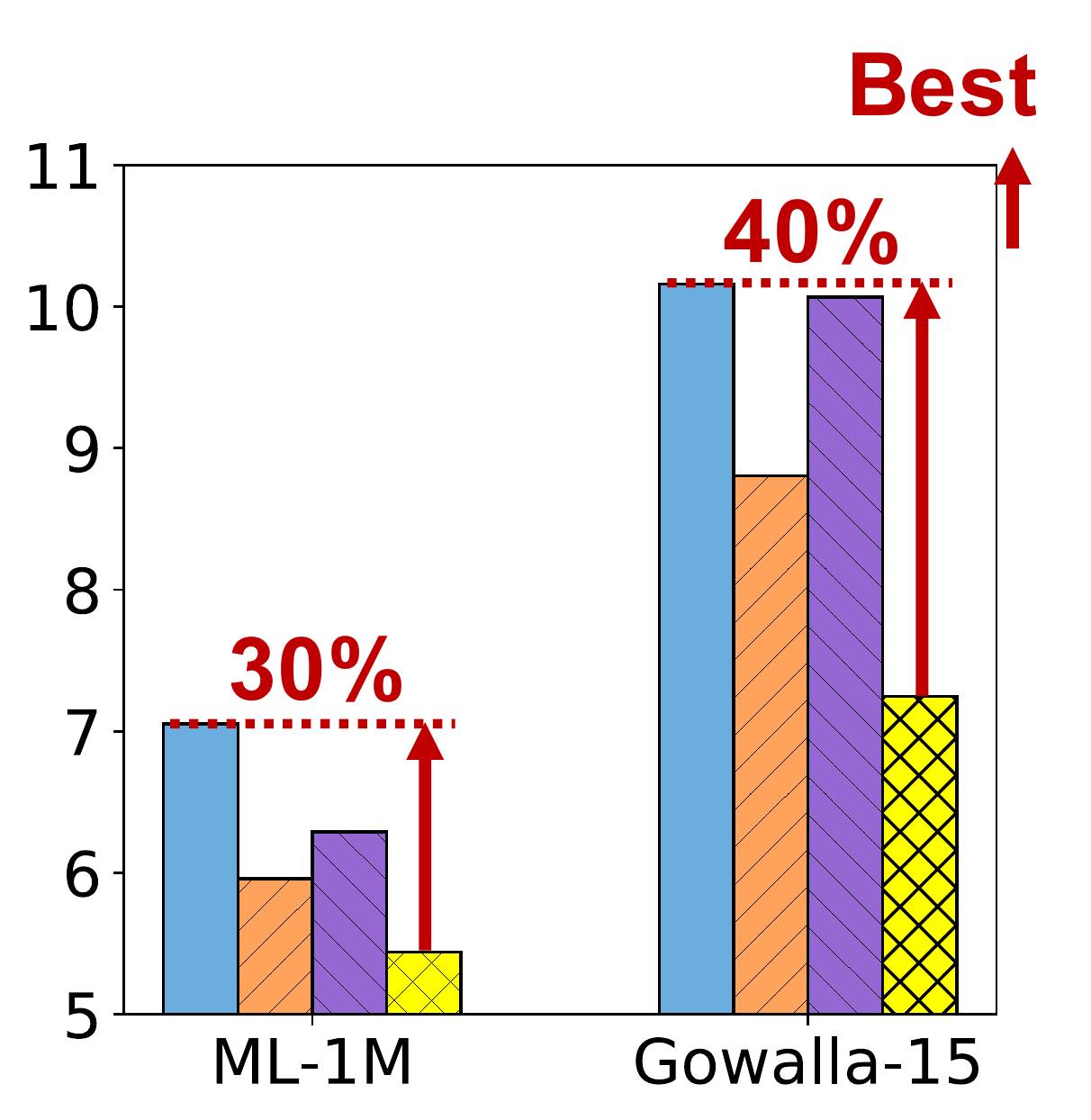}\label{fig:ablation_ml1m_ent}}
	 \subfigure[Negative Gini index]{\includegraphics[width=0.335\columnwidth]{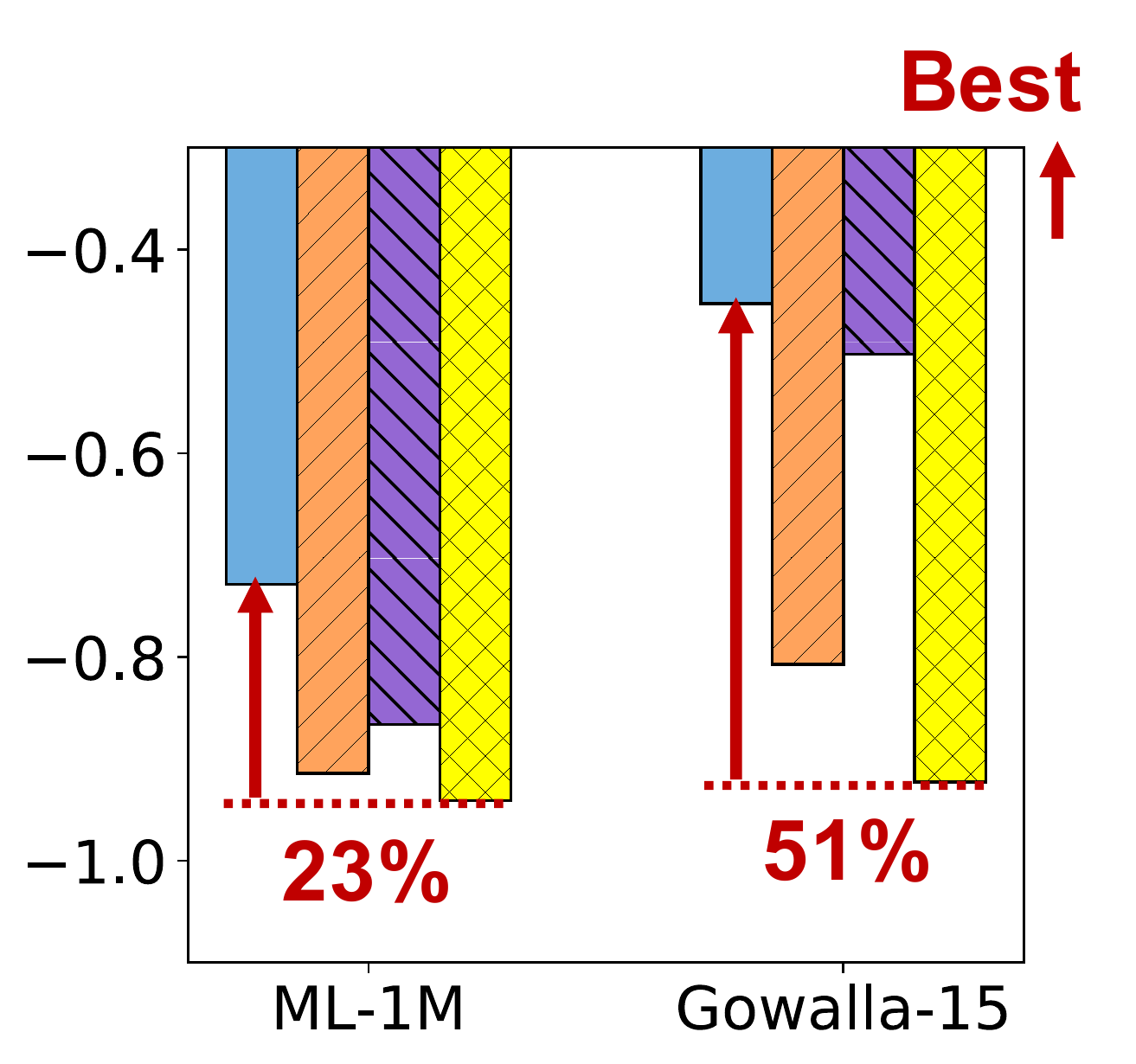}\label{fig:ablation_ml1m_gini}}
	 \vspace{-4mm}
	\caption{
		The change of diversity metrics compared to MF when nDCG is decreased by five percent.
		\method shows the best diversities with both coverage regularizer and skewness regularizer.
	}
	\label{FIG:ablation}
\end{figure}

We verify the impact of coverage regularizer and skewness regularizer. 
We compare improvements in aggregate-level diversity of three models: \method, \method-$Reg_{skew}$, and \method-$Reg_{cov}$ compared to MF in the ML-1M and Gowalla-15 datasets.
\method is our proposed method which utilizes both coverage regularizer and skewness regularizer.
\method-$Reg_{skew}$ and \method-$Reg_{cov}$ do not use the skewness regularizer and the coverage regularizer, respectively, in training process.
For fair comparison of performance, we measure the aggregate-level diversities of models when their nDCGs are equally dropped by five percent compared to MF since models achieve higher diversity with lower accuracy.

Figure~\ref{FIG:ablation} shows the results.
We have two main observations.
First, \method shows the best results in all metrics.
This shows that the skewness regularizer combined with the coverage regularizer supports maximizing entropy and minimizing Gini index as shown in Theorem \ref{thm:frequency}.
Second, \method-$Reg_{cov}$ shows worse results than \method-$Reg_{skew}$.
This is because the skewness regularizer alone does not consider the balance between items at the aggregate level,
while the coverage regularizer adjusts the distribution of recommended items at the aggregate level as shown in Theorem \ref{thm:coverage}.


\subsection{Training Algorithm (Q3)}
\label{subsec:exp_algs}

\begin{figure}[h]
	\centering
	\subfigure{\includegraphics[width=0.8\columnwidth]{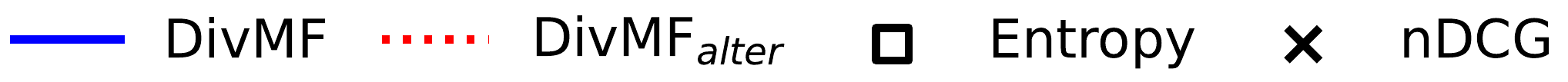}}\\\vspace{-3mm}
	\setcounter{subfigure}{0}
	\subfigure{\includegraphics[width=0.8\columnwidth]{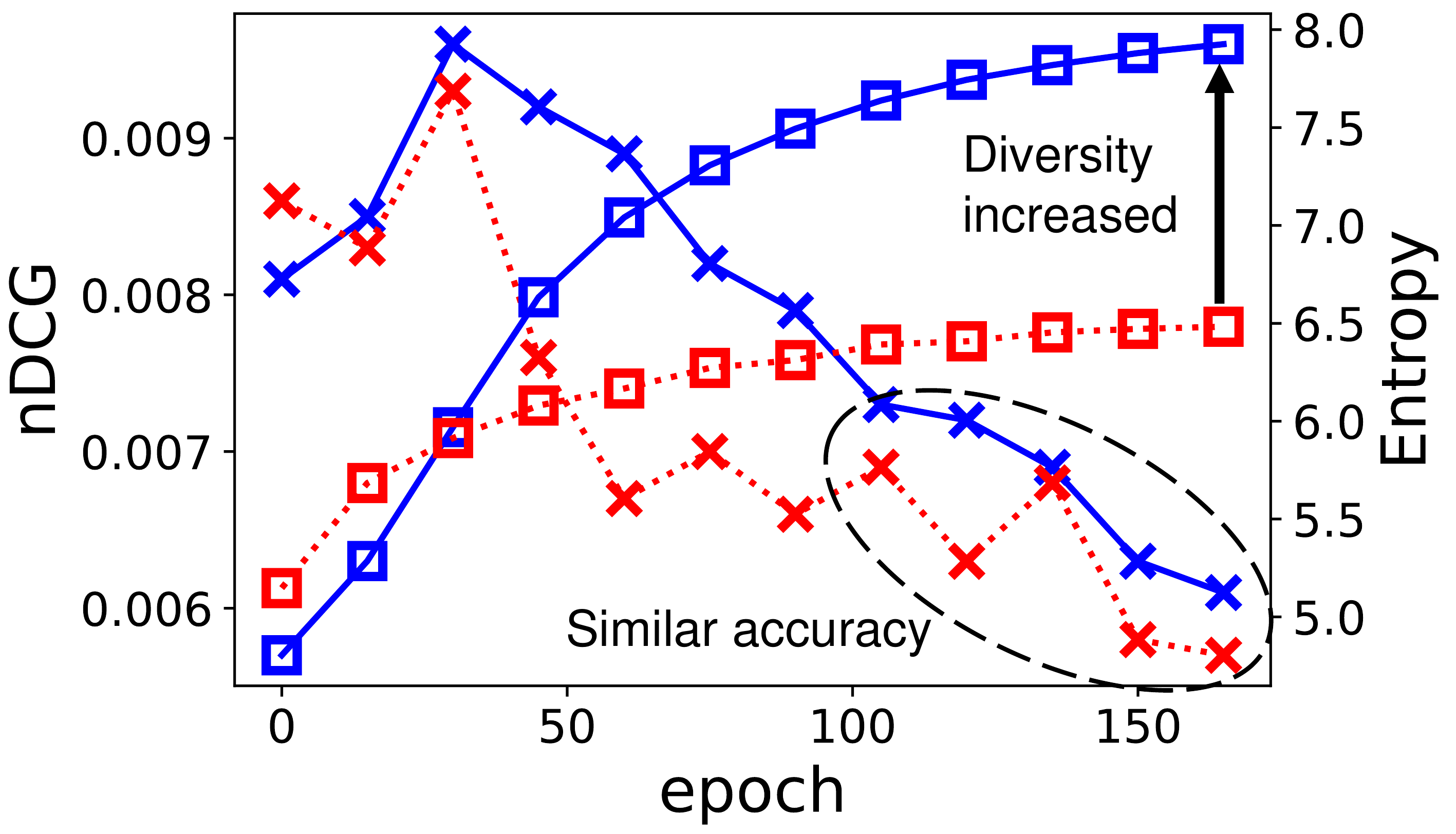}\label{fig:algs_ml}}
	\vspace{-4mm}
	\caption{
		Change of nDCG and entropy of \method and \method$_{alter}$ during training.
		\method achieves a significantly higher diversity than \method$_{alter}$ when their levels of accuracy are similar since \method avoids being trapped in bad local optima.
	}
	\label{FIG:algs}
\end{figure}

We verify the effectiveness of our algorithm (Algorithm~\ref{al:training}) to train the model in Figure~\ref{FIG:algs}.
We compare training processes of \method and \method$_{alter}$ on ML-1M dataset.
\method$_{alter}$ alternately optimizes the accuracy loss and the diversity regularizer unlike \method,
which optimizes the diversity regularizer after optimizing the accuracy loss.

Figure~\ref{FIG:algs} shows that \method with Algorithm~\ref{al:training} successfully prevents the model from being trapped in bad local optima.
The increased entropy of \method$_{alter}$ until convergence is significantly smaller than that of \method, while the accuracies of the two models are at a similar level.

\subsection{Unmasking Mechanism (Q4)}
\label{subsec:unmask}
To find optimal unmasking policy for \method,
we compare three different unmasking schemes: no unmasking scheme, unmasking top-$k+n$ scheme, and unmasking random $n$ scheme.
$n$ is a hyperparameter that controls the number of additional unmasked items.
Unmasking top-$k+n$ scheme unmasks total $k+n$ items with high prediction scores. 
Unmasking random $n$ scheme unmasks randomly selected $n$ items. 
We set $n=100$ for each scheme since this setting shows the best overall performance.

Figure~\ref{FIG:unmask_scheme} shows the trade-off curves of Top-$5$ recommendation in each unmasking scheme of \method on ML-1M dataset.
We have two observations.
First, \method fails to achieve high aggregate-level diversity without unmasking mechanism.
The model with no unmasking scheme reaches the limit on diversity during training while other models with an unmasking mechanism further increase both coverage and negative Gini index.
Second, it is better to unmask high scored items than random items to achieve better coverage.
The reason is that there is a higher chance to find proper items to recommend from the top of the recommendation list than from random items.

\begin{figure}[t]
	\centering
	\subfigure{\includegraphics[width=0.6\columnwidth]{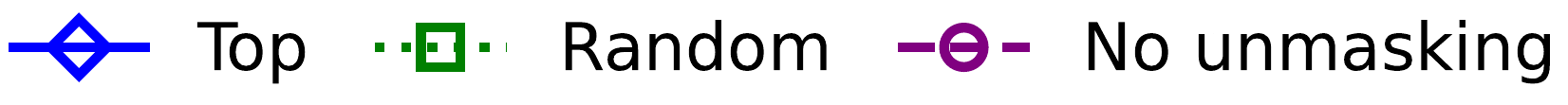}}\\\vspace{-2mm}
	\setcounter{subfigure}{0}
	\subfigure[Coverage]{\includegraphics[width=0.48\columnwidth]{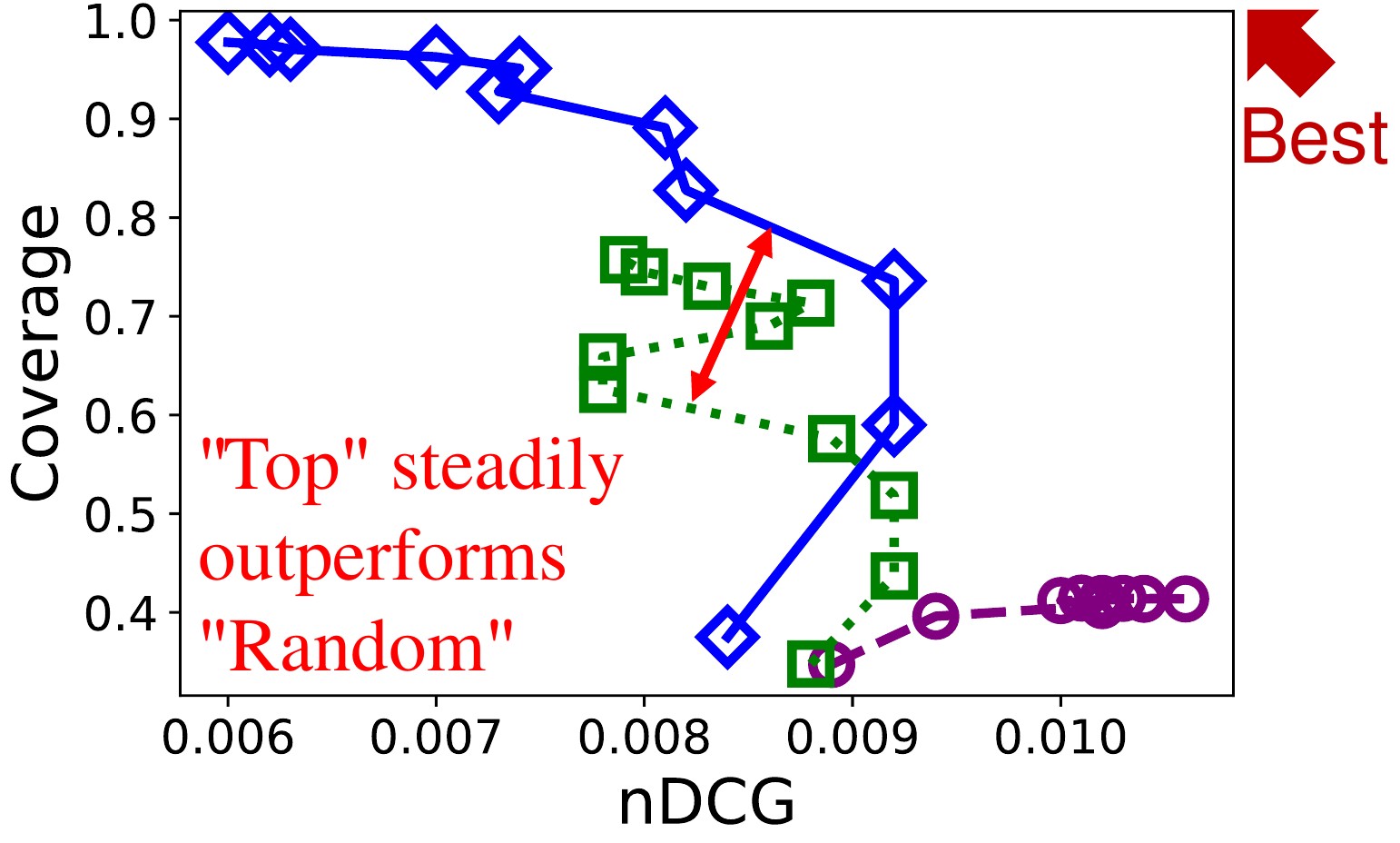}\label{fig:unmask_ml1m_cov}}
	\subfigure[Negative Gini index]{\includegraphics[width=0.48\columnwidth]{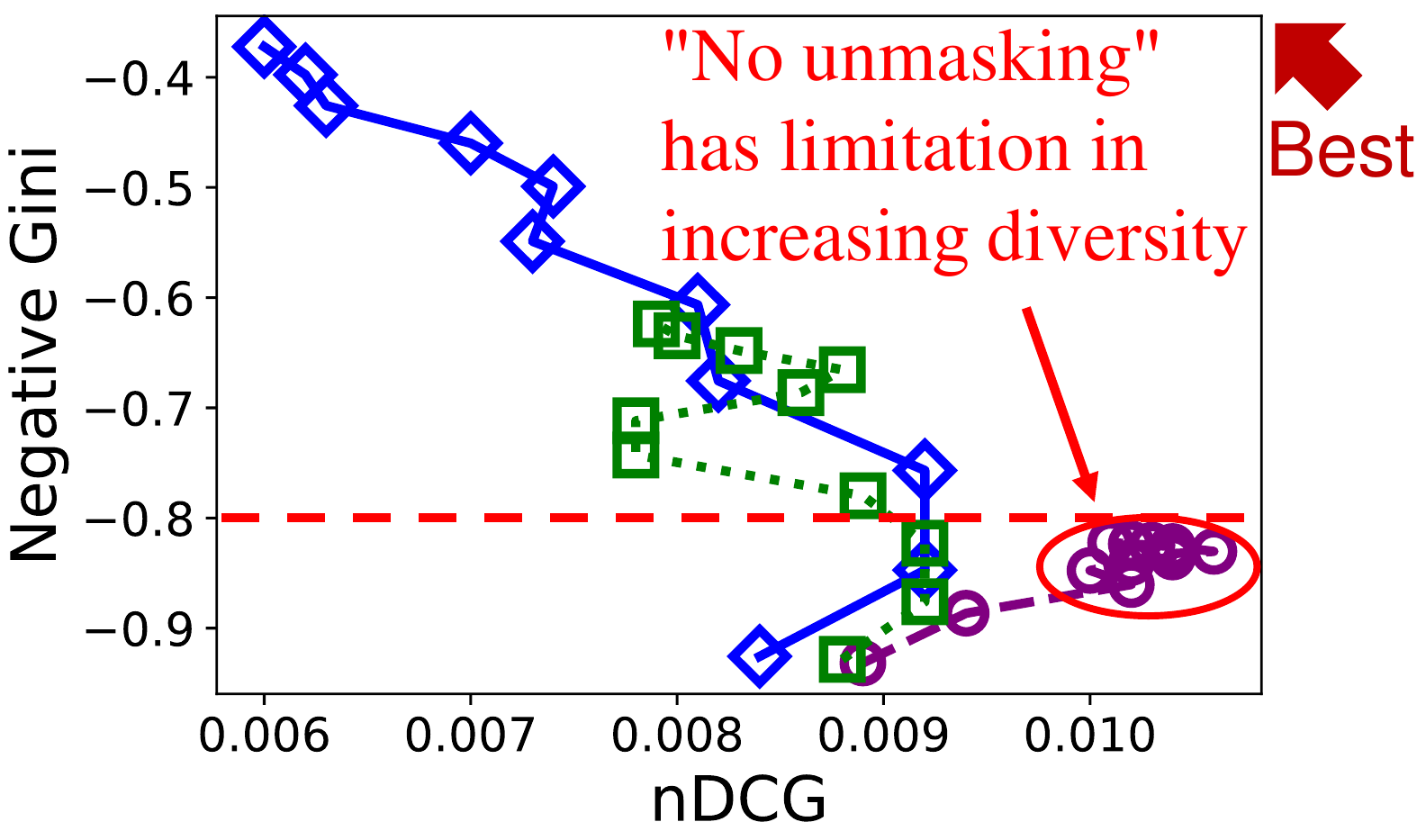}\label{fig:unmask_ml1m_ent}}
		 \vspace{-4mm}
	\caption{
		Trade-off curves of different unmasking policies.
		It is better to unmask high scored items rather than random items or no items.
	}
	\label{FIG:unmask_scheme}
\end{figure}

\subsection{Mini-Batch Learning (Q5)}
\label{subsec:batch_exp}

\begin{figure}[h]
	\vspace{-5mm}
	\centering
	\subfigure{\includegraphics[width=0.75\columnwidth]{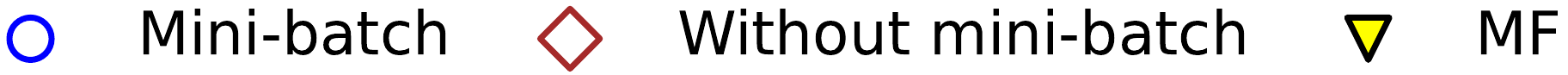}}\\\vspace{-3mm}
	\setcounter{subfigure}{0}
	\subfigure[Gowalla-15]{\includegraphics[width=0.49\columnwidth]{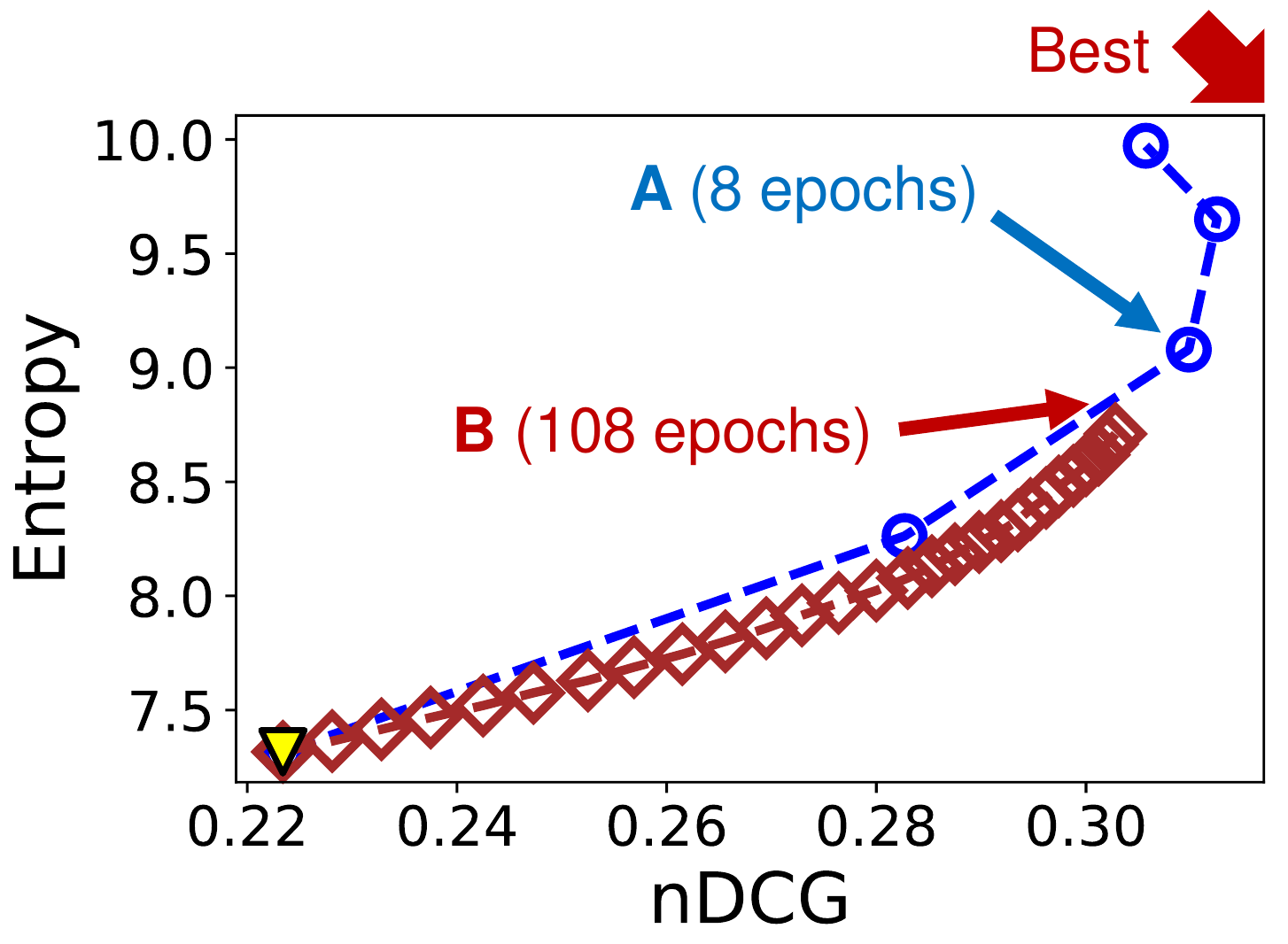}\label{fig:batch_gowa}}
	\subfigure[Yelp-20]{\includegraphics[width=0.49\columnwidth]{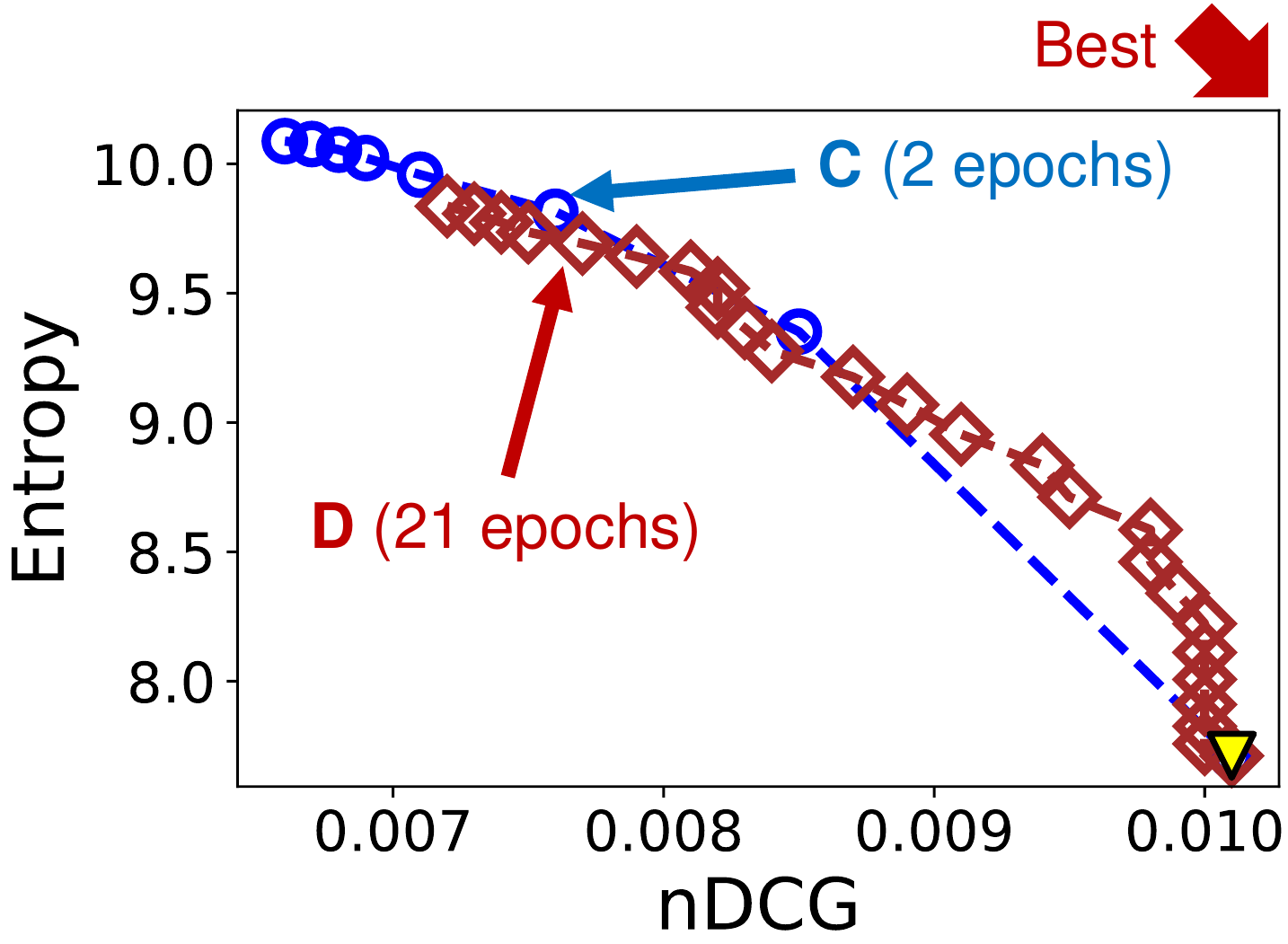}\label{fig:batch_yelp}}
	\vspace{-4mm}
	\caption{
		Trade-off curves of \method with and without the mini-batch learning technique in Gowalla-15 and Yelp-20 datasets.
	\method with the mini-batch learning requires less than 10\% of epochs required by that without the mini-batch learning to achieve comparable performance, enabling fast training.
	}
	\label{FIG:batch_plot}
\end{figure}

We compare results of \method with and without the mini-batch learning technique (Section~\ref{subsec:batch}) in large datasets such as Gowalla-15 (15-core preprocessed version of Gowalla) and Yelp-20 (a 20-core preprocessed version of Yelp) to confirm that the mini-batch technique does not decrease the recommendation quality of \method while shortening the training time.
\method without the mini-batch technique is trained with CPU since the model is too large to be stored on the GPU.
We use Yelp-20 instead of Yelp-15 to fit the model in memory capacity.
Yelp-20 dataset contains 2,438,708 instances of 41,774 users and 30,037 items.
We unmask 500 items in Gowalla-15 and 100 items in Yelp-20 dataset.
We use batch size $r_b=5000$ and $c_b=5000$ for both datasets.

Figure~\ref{FIG:batch_plot} shows the accuracy-diversity trade-off curves of Top-$5$ recommendations of \method with and without the mini-batch learning technique.
All curves start from the result of MF which is marked as a yellow triangle.
We increase the $n_{ep}$ by 4 for every dot in Gowalla-15 and by 1 for every dot in Yelp-20.
We have two observations.
First, mini-batch learning technique enables fast training.
\method with the mini-batch learning technique requires less than 10\% of epochs to achieve similar performance as that of \method without the mini-batch learning. 
In Yelp-20 dataset, the performance of \method with the mini-batch learning technique at epoch 2 (dot C) and that of \method without the mini-batch learning technique at epoch 21 (dot D) are similar. 
Second, mini-batch learning also improves diversity of recommendation.
Note that in Gowalla-15 dataset, the entropy of recommendation result by \method with the mini-batch learning technique is drastically improved compared to that of without mini-batch learning.
In summary, the mini-batch learning technique accelerates the training process of \method and additionally improves the performance.

%% file: 050related.tex
\section{Related Works}
\label{sec:related}
We summarize previous works related to aggregately diversified recommendation: individually diversified recommendation, fair recommendation, popularity debiased recommendation, and capacity constrained recommendation.

\textbf{Individually diversified recommendation.}
Individually diversified recommendation recommends diversified items to each user~\cite{WuLMZT19}.
Maximizing individual diversity can maximize item novelty in each user's view, but it may recommend already known items in overall recommendation list for all users.
Thus, maximizing individual-level diversity does not guarantee the improvement in aggregate-level diversity~\cite{AdomaviciusK12}.

\textbf{Fair recommendation.}
Fair recommendation aims to design an algorithm that makes fair predictions devoid of discrimination~\cite{Gajane17}.
There are several categories in fairness recommendation based on the definition of fairness.
For instance, group fairness recommendation considers that a model is fair when items from different groups have similar chance of being recommended~\cite{pager2008sociology, zehlike2017fa, ekstrand2021exploring}.
On the other hand, individual fairness recommendation considers that a model is fair if items with similar attributes are equally recommended~\cite{patro2020fairrec}.
Aggregately diversified recommendation
does not require any groups or attributes to be defined,
which is the main difference compared to fair recommendation.

\textbf{Popularity debiased recommendation.}
Popularity debiased recommendation aims to improve the quality of recommendation for long-tail items.
Traditional recommender systems tend to show poor accuracy in recommending infrequently appeared items compared to frequently appeared items because of the skewness in dataset~\cite{steck2011item}.
There are researches to eliminate the popularity bias to achieve high accuracy in recommending long-tail items as well as popular items~\cite{wei2021model, ferraro2019music}.
Aggregately diversified recommendation focuses on increasing the frequency of long tail items in recommendation lists in addition to recommending them accurately,
which is the main difference from popularity debiased recommendation.

\textbf{Capacity constrained recommendation.}
Capacity constrained recommendation aims to provide good recommendations while taking the limited resources into consideration~\cite{ChristakopoulouKB17}.
Recommendations are distributed to more diverse items when some popular items do not have enough capacity, hence aggregated diversity of recommendation may be achieved.
 However, if there is enough capacity to most items, long-tailed items remain not recommended.
 Also, the experiment for capacity constrained recommendation is limited because there have been no datasets with capacity of items listed.
 Previous works set capacities of items randomly, and thus do not represent the real-world scenario~\cite{ChristakopoulouKB17}.
 Aggregately diversified recommendation
 does not require the capacity of items as input data,
 which is the main difference compared to capacity constrained recommendation.

%% file: 060conclusion.tex
\section{Conclusion}
\label{sec:conclusion}

We propose \method, a matrix factorization method for aggregately diversified recommendation, which maximizes aggregate diversity while sacrificing minimal accuracy.
\method regularizes an MF model with coverage and skewness regularizers which maximizes the coverage and entropy, respectively, of top-$k$ recommendation in the training process.
Furthermore, \method effectively optimizes the model with an unmasking mechanism 
and carefully designed mini-batch learning. 
Experiments on five real-world datasets show that \method achieves the state-of-the-art performance in aggregately diversified recommendation outperforming the best competitor with
up to 34.7\% reduced Gini index in the similar level of accuracy and up to 27.6\% higher nDCG in the similar level of diversity.
Future works include extending \method to other fields such as individually diversified recommendation and fair recommendation. 